\documentclass[12pt]{article}
\usepackage{setspace}
\usepackage{amsmath,amssymb,mathrsfs}
\usepackage{color}
\usepackage{hyperref}
\usepackage{graphicx,epsfig}
\usepackage{natbib} 
\usepackage{xcolor}
\renewcommand{\emph}[1]{\textit{#1}}
\numberwithin{equation}{section}

\newcommand{\beq}{\begin{equation}}
\newcommand{\eeq}[1]{\label{#1}\end{equation}}
\newcommand{\bea}{\begin{eqnarray}}
\newcommand{\eea}[1]{\label{#1}\end{eqnarray}}

\usepackage{enumitem}
\usepackage{mathtools}
\usepackage{bbm}
\usepackage{amsthm}
\usepackage{amsmath,amssymb,amsfonts}
\usepackage{appendix}
\usepackage{diagbox}

\DeclareMathOperator{\im}{i}

\newtheorem{remark}{Remark}
\newtheorem{proposition}{Proposition}

\def\Cl3{{\mathcal{C}\ell}_{3,0}}
\def\Cl{{\mathcal{C}\ell}_{0,m}}

\def\Cl{{\mathcal{C}\ell}_{0,m}}

\newcommand{\seqnum}[1]{\href{https://oeis.org/#1}{\rm \underline{#1}}}

\begin{document}
\begin{titlepage}
\begin{center}
\vskip 4 cm

{\Large \bf Clifford-Appell formulation of a Dirac-type Kronig-Penney model in condensed matter physics}

\vskip 1 cm

{\large \bf L. Aceto$^{a,}$~\footnote{lidia.aceto@uniupo.it
}, 
\large {\bf P.A. Grassi}$^{a,b,}$~\footnote{{pietro.grassi@uniupo.it}}, and   
\large {\bf H. Malonek}$^{c,}$~\footnote{{
hrmalon@ua.pt}}}

\vskip .75 cm

{ $^{(a)}$ \it Dipartimento di Scienze e Innovazione Tecnologica (DiSIT),\\
Universit\`a del Piemonte Orientale, Viale T. Michel, 11, 15121 Alessandria, Italy}

 {$^{(b)}$
 \it INFN, Sezione di Torino and  Centro Fermi} 
 {\it Via P. Giuria, 1, 10125, Torino, Italy}

{$^{(c)}$ \it CIDMA  and Departamento de Matemática
Campus Universitário de Santiago, 3810-193, 
Aveiro, Portugal}

\end{center}

\begin{abstract}
We investigate the stationary Dirac equation associated with a generalized Kronig-Penney model in \((N+1)\)-dimensional Clifford analysis. Away from the interaction sites, the free Dirac system is reformulated by introducing suitable generalized Cauchy-Riemann operators, leading to a coupled first-order system for two Clifford-valued fields. By means of characteristic hypercomplex variables, this system is reduced to a pair of decoupled generalized Helmholtz equations. To construct explicit solutions, we introduce a bivariate Clifford-Appell polynomial basis adapted to the generalized Cauchy-Riemann operators. The Appell expansions transform the differential problem into an algebraic recurrence for the expansion coefficients, yielding explicit series representations of the solutions. The proposed framework provides a unified Clifford-analytic formulation of the free propagation problem and establishes an explicit connection between generalized Helmholtz equations, and Appell polynomial systems in hypercomplex function theory. 
\end{abstract}

\end{titlepage}

%\maketitle

%%%%%%%%%%%%%%%%%%%%%%%%%%%%%%%%%%%%%%%%%%%%%%%%%%%%%%%%%%%%%%%%%%%%%%%%%%%%%%%%%%%%%%%%%%%%%%%%%%%%%%%%%%%%%%%%%%%%%%%%%%%%%%%%%%%%

\section[1]{Introduction} \label{sec:Cliff]}

The one-dimensional quantum scattering problem of a particle with positive energy $E$ interacting with a strongly localized potential has recently been investigated in \cite{FGS-2025}. In this setting, the interaction is modeled by a finite array of Dirac delta distributions, leading to the finite Kronig-Penney model, a paradigmatic system in condensed matter physics. The dynamics are governed by the stationary Schrödinger equation
\begin{equation}
\label{SCA}
-\frac{\hbar^2}{2m}\psi''(x)
+\lambda\sum_{n=1}^{N}\delta(x-nL)\psi(x)
=
E\psi(x),
\end{equation}
where the potential consists of $N$ equidistant Dirac delta interactions located at $x=nL$, with $L$ denoting the lattice spacing. The coupling constant $\lambda$ determines the physical regime of the system: $\lambda>0$ corresponds to an array of repulsive scatterers, whereas $\lambda<0$ describes a sequence of attractive wells capable of supporting bound states.

In \cite{FGS-2025}, a rigorous analytical framework was developed to obtain a closed-form expression for the transfer matrix of this finite Kronig-Penney model. Such a representation provides a systematic characterization of the scattering and transmission properties of the system and allows a detailed investigation of its resonant structure.

A noteworthy feature revealed by this analysis is the emergence of specific non-symmetric triangular arrays of integers. The elementary examples arising in this context are listed in Table~\ref{table1}. These configurations were previously studied from a combinatorial perspective in \cite{Cacao-2023}, where they were classified in connection with corresponding entries in the \textit{On-Line Encyclopedia of Integer Sequences} (OEIS).
\begin{table}[ht]
\centering
\caption{Three integer triangular arrays showing the first four rows
($k=0,1,2,3$ and $0\leq s\leq k$) considered in \cite[Table~1]{Cacao-2023}.}
\label{table1}

\renewcommand{\arraystretch}{1.2}

\begin{tabular}{l@{\qquad\qquad}l}
\hline
\\[-1ex]

$\displaystyle
\begin{array}{rrrr}
1\\
2&1\\
3&2&1\\
4&3&2&1
\end{array}
$
&
\seqnum{A004736}
\\[3ex]

\hline
\\[-1ex]

$\displaystyle
\begin{array}{rrrr}
1\\
3&2\\
6&6&3\\
10&12&9&4
\end{array}
$
&
\seqnum{A104633}
\\[3ex]

\hline
\\[-1ex]

$\displaystyle
\begin{array}{rrrr}
1\\
4&3\\
10&12&6\\
20&30&24&10
\end{array}
$
&
\seqnum{A103252}
\\[2ex]

\hline

\end{tabular}

\end{table} 
The mathematical background of these arrays can be traced back to the hypercomplex construction introduced in \cite{Malonek-Falcao-2006}, where non-symmetric triangular arrays first appeared naturally in the study of quasi-conformal mappings through hypercomplex monogenic polynomials.
Subsequent investigations have uncovered further analytic and combinatorial properties of these structures, highlighting their connections with Clifford analysis and discrete Clifford analysis; see, for example, \cite{Bock-2014,BSL-1992,CMFT-2026,CMFT2018,Eelbode-2012,Hitzer-2017}.

The occurrence of such arrays in a quantum scattering framework suggests the existence of an underlying algebraic structure that remains hidden in the original Schrödinger formulation. The present work pursues this line of research by investigating the unexpected manifestation of these arrays in a quantum-mechanical setting. Our aim is to show that the patterns observed in the finite Kronig–Penney model originate from the Clifford-algebraic structure underlying a suitable first-order reformulation of the Schrödinger equation. In particular, by associating a Dirac-type equation with the scattering problem, we demonstrate how the dimensional parameter of the Clifford framework is encoded in the patterns emerging from the transfer-matrix formulation.

While the present paper focuses on the physical and algebraic origin of these structures, the explicit construction of the non-symmetric number triangles arising from the principal transfer matrix, together with their recursive properties and combinatorial interpretation, is presented in a separate contribution \cite{AcetoGrassiMalonek2026}.

The main result of the present approach is the identification of these arrays as genuine manifestations of the hidden Clifford-algebraic structure underlying the Dirac-type formulation, rather than as accidental numerical features of the scattering coefficients. Although Eq.~\eqref{SCA} describes a one-dimensional quantum system, its first-order reformulation naturally leads to higher-dimensional Clifford representations, whose algebraic properties are encoded by generalized hypercomplex Appell polynomials. Finally, the Dirac-type equation can be embedded into a Dirac equation by introducing auxiliary space coordinates on which the corresponding $\gamma$-matrices act.

This connection also places the present analysis within a broader context in which similar arithmetic structures arise in quantum theory. In particular, works published in 2019 and 2021 \cite{Wigner-2019,CommentWigner-2021} identified related triangular configurations, interpreted as Wigner number trapezoids, in the framework of Wigner rotation matrices and angular momentum theory. Such parallels suggest that these combinatorial structures reflect recurring algebraic patterns underlying different mathematical formulations of quantum phenomena.

The paper is organized as follows. Section~\ref{sec:2}  provides the analytical background on hypercomplex function theory. Section~\ref{sec:3} is devoted to the matrix representations, projection operators, recursive constructions of the \(\gamma\)-matrices, and their algebraic properties. Section~\ref{sec:4} establishes the connection between the finite Kronig-Penney model and the associated Dirac-type equation. Section~\ref{sec:5} presents the Clifford-based formulation of the exact solutions and the corresponding scattering properties. Section~\ref{sec:6} introduces the generalized Cauchy-Riemann framework leading to the associated Helmholtz equations. Section~\ref{sec:7} develops the Appell expansions for Clifford-valued Helmholtz equations and discusses their role in the representation of the transfer structure. Finally, Section~\ref{sec:FINE} contains concluding remarks and perspectives for future developments.

%%%%%%%%%%%%%%%%%%%%%%%%%%%%%%%%%%%%%%%%%%%%%%%%%%%%%%%%%%%%%%%%%%%%%%%%%%%%%%%%%%%%%%%%%%%%%%%%%%%%%%%%%%%%%%%%%%%%%%%%%%%%%%%%%%%%%%%%%%%%%%%%%%%%%%%%%%%%
\section[2]{Hypercomplex function theory and Appell polynomials} \label{sec:2}

To formulate the scattering problem in arbitrary spatial dimensions, the complex analytic framework must be replaced by its higher-dimensional extension. This extension is provided by Clifford analysis \cite{BrackxDelangheSommen-1982}, which is based on the real Clifford algebra $\Cl.$  The algebra  $\Cl$ is generated by an orthonormal basis $\{e_1, e_2,\dots, e_m\}$ of $\mathbb{R}^m$ satisfying the anti-commutation relations
\begin{equation*} \label{clifford-mult}
e_k e_l + e_l e_k = -2\delta_{kl}, \quad k,l=1,2,\dots,m,
\end{equation*}
where $\delta_{kl}$ denotes the Kronecker delta. The set $\{e_A:A\subseteq \{ 1,\dots,m\}\}$ with
\begin{equation*}
	e_A=e_{h_1}e_{h_2}\cdots e_{h_r}, \ 1\leq h_1< \cdots < h_r\le m,\ e_{\emptyset}=e_0=1,
\end{equation*}
forms a basis of the $2^m$-dimensional Clifford algebra $\Cl$ over $\mathbb{R}.$  In this framework, the elements of $\mathbb{R}^{m+1}$ are identified with paravectors of the Clifford algebra $\Cl.$ In particular, for
\[
x = (x_0, x_1, \dots, x_m) \in \mathbb{R}^{m+1},
\]
the corresponding paravector is defined as
\begin{equation*}
x = x_0 + \mathbf{x} = x_0 + \sum_{k=1}^m x_k e_k,
\end{equation*}
where  $x_0$ and  $\mathbf{x}$ denote the scalar and vector parts, respectively. The Clifford conjugation and the norm of a paravector are given by
\[
\bar{x} = x_0 - \mathbf{x}, \qquad |x| = \sqrt{x\bar{x}}.
\]
This norm provides the Euclidean metric structure on the space of paravectors.  A $\Cl$-valued function defined on an open set  $\Omega\subset \mathbb{R}^{m+1}$ admits the representation
\[
f(x)=\sum_A f_A(x)e_A,
\]
where $f_A:\Omega \rightarrow \mathbb{R}$ are real-valued functions. The higher-dimensional counterpart of holomorphic functions is provided by left- and right-monogenic functions, defined as the kernel of the generalized Cauchy–Riemann operator
\begin{equation} \label{dbaroperator}
\overline{\partial} := \frac{1}{2} \left( \partial_0 + \partial_{\mathbf{x}} \right),
\end{equation}
where the action is understood from the left (respectively, from the right) on $f.$ Here $\partial_0 = \frac{\partial}{\partial x_0}$   and $\partial_{\mathbf{x}}$ denotes the Dirac operator associated with the spatial variables
\begin{equation} \label{dx}
\partial_{\mathbf{x}} = \sum_{k=1}^m e_k \frac{\partial}{\partial x_k}.
\end{equation}
For monogenic functions, the hypercomplex derivative is defined through the action of the conjugate operator
\begin{equation} \label{del-operator}
\partial := \frac{1}{2}(\partial_0 - \partial_{\mathbf{x}}).
\end{equation}
This derivative coincides with the hypercomplex (areolar) derivative of a monogenic function in the sense of Pompeiu, as shown in \cite{Gurlebeck-1999}. For $m>1,$ the existence of this derivative is a necessary and sufficient condition for monogenicity. For the paravector $x,$ the action of the generalized Cauchy–Riemann operator \eqref{dbaroperator} yields  
\begin{equation}\label{nonmonogenic}
	\overline{\partial} x = x \overline{\partial} =  \frac{1}{2}(\partial_0  x_0 + \partial_{\mathbf{x}} \mathbf{x}) = \frac{1}{2}(1-m). 
\end{equation}
Similarly, applying the conjugate operator \eqref{del-operator} gives
\begin{equation}\label{nonderivada}
	{\partial} x = x {\partial}  = \frac{1}{2}(\partial_0  x_0 - \partial_{\mathbf{x}} \mathbf{x}) = \frac{1}{2}(1+m). 
\end{equation}
Together, these identities show that only in the complex case $m=1$ the paravector $x$ is monogenic and its hypercomplex derivative is equal to $1$. A direct computation also shows that the integer powers of $x$ are not monogenic in general (see \cite{Cacao-2024}). Consequently, monogenic polynomials cannot be obtained by a straightforward extension of ordinary polynomials in the variable $x.$ This observation motivates the introduction of the hypercomplex Appell polynomials with respect to $\partial,$ denoted by $\{ \mathcal{P}_k^m(x,\bar{x}) \}_{k \ge 0}.$  These polynomials have been explicitly constructed in \cite{Falcao-2006, Malonek-Falcao-2006} and are given by
\begin{equation} \label{appell-pol}
\mathcal{P}_k^m(x,\bar{x}) = \sum_{s=0}^k T_s^k(m) x^{k-s} \bar{x}^s.
\end{equation}
Here, the coefficients $T_s^k(m)$ can be expressed in terms of the Pochhammer symbol
$$(a)_r: =\frac{\Gamma(a+r)}{\Gamma(a)} = a(a+1)(a+2)\cdots (a+r-1), \quad r \ge 1, \qquad (a)_0 = 1,$$ 
as
\begin{equation} \label{num-tria}
			T_s^k(m) = \frac{k!}{(m)_k} \frac{\left( \frac{m+1}{2} \right)_{{k-s}}}{(k-s)!} \frac{\left( \frac{m-1}{2} \right)_{s}}{s!}, \qquad m\ge1, \quad k \ge 0, \quad s=0,1,\dots, k.
		\end{equation}
The presence of both $x$ and $\bar x$ in \eqref{appell-pol} reflects a fundamental difference with respect to the complex case. Indeed, due to \eqref{nonmonogenic} and \eqref{nonderivada}, monogenic functions in the hypercomplex setting do not generally admit a power series expansion solely in terms of $x.$ The coefficients $T_s^k(m)$ are chosen such that the sequence satisfies the Appell property
\begin{equation*}
	\partial \mathcal{P}_k^m(x,\bar x) = k \mathcal{P}_{k-1}^m(x,\bar x), \quad  \quad \mathcal{P}_0^m(x,\bar x) \equiv 1.
\end{equation*}
They encode the dependence on the dimension $m$ and determine the combination of $x$ and  $\bar x$ required to preserve monogenicity.

Several relevant properties of the triangular coefficients $T_s^k(m)$
highlight their fundamental role in the construction of hypercomplex
Appell polynomials. For each fixed homogeneous degree $k$, these
coefficients form a partition of unity \cite[Theorem 3.7]{Falcao-2012}.
Moreover, their alternating combination gives rise to the generalized
Vietoris numbers $c_k(m)$, which have been investigated in \cite{CMF-2019}. More precisely, these relations take the form
\begin{equation}\label{unit-sum}
	\sum_{s=0}^k T_s^k(m)=1,
	\qquad
	c_k(m)=\sum_{s=0}^k(-1)^sT_s^k(m),
	\qquad k=0,1,\dots .
\end{equation}
These identities show that the coefficients $T_s^k(m)$ contain
combinatorial information beyond their direct role in the definition of
hypercomplex Appell polynomials.

For $m=1$, the triangular array reduces to a single non-zero entry,
since
\[
T_0^k(1)=1,\qquad T_s^k(1)=0,\quad s>0.
\]
This case corresponds to the Clifford algebra
${\mathcal C\ell}_{0,1}\cong\mathbb{C}$ and recovers the classical
complex Appell system of monomials
\[
x^k\cong z^k=(x+iy)^k,
\]
which satisfies
\[
(z^k)'=kz^{k-1},\qquad z^0=1.
\]
Thus, the standard complex polynomial framework appears as the one-dimensional realization of the generalized hypercomplex construction.

For $m\geq2$, the coefficients $T_s^k(m)$ describe a genuinely
higher-dimensional structure. Although the alternating sums are related to the elements  $c_k(m)$ of the generalized Vietoris sequences (see 
\eqref{unit-sum}), the main object of interest here is the full
two-indexed family $T_s^k(m)$, which encodes the dependence of the
hypercomplex Appell polynomials on the dimension parameter. In
particular, unlike the one-dimensional coefficient sequences associated
with Taylor expansions of holomorphic functions, the representation of
monogenic functions through the variables $x$ and $\bar{x}$ naturally
requires a bidimensional coefficient array. Indeed, the indices $k$ and $s$ have distinct algebraic meanings: the first one
labels the homogeneous degree of the Appell polynomial, while the second
one counts the contribution of the conjugate variable $\bar{x}$. This
structure leads to the triangular arrangement displayed in
Table~\ref{table2}. Although this representation resembles the classical
Pascal-Tartaglia triangle, the coefficients $T_s^k(m)$ are governed by
the non-commutative Clifford algebra setting and therefore require a
different geometric interpretation.
\begin{table}[ht]
\centering
\caption{The first four rows of the triangular array $T_s^k(m)$.}
\label{table2}
\renewcommand{\arraystretch}{1.8}
\setlength{\tabcolsep}{8pt}
\[
\begin{array}{c|cccc}
k\backslash s & 0 & 1 & 2 & 3\\
\hline
0 &
1
\\[1ex]
1 &
\dfrac{m+1}{2m} &
\dfrac{m-1}{2m}
\\[1ex]
2 &
\dfrac{m+3}{4m} &
\dfrac{2(m-1)}{4m} &
\dfrac{m-1}{4m}
\\[1ex]
3 &
\dfrac{(m+5)(m+3)}{8m(m+2)} &
\dfrac{3(m+3)(m-1)}{8m(m+2)} &
\dfrac{3(m+1)(m-1)}{8m(m+2)} &
\dfrac{(m+3)(m-1)}{8m(m+2)}
\end{array}
\]
\end{table}

Higher-dimensional extensions of this triangular structure provide a
natural link with multidimensional hypercomplex polynomial systems.
In particular, the combinatorial organization of the coefficients
suggests higher-dimensional analogues of Pascal-type structures
\cite{CCFMT-2025}, which reflect the increasing complexity of
hypercomplex Appell systems beyond the classical complex case.

While these structures have been mainly studied within hypercomplex
function theory, the following sections show that they also emerge
naturally in the transfer-matrix formulation of the quantum scattering
problem. In this context, the correspondence is not immediate: the
triangular organization of the Appell coefficients is modified by the
number of barriers $N$ in the finite Kronig--Penney model and by the
dimension-dependent scaling contained in the coefficients
\eqref{num-tria}. By reformulating the scattering problem through a
Dirac-type equation, we demonstrate that the appearance of these
triangular arrays is intrinsically connected with the underlying
structure of quantum propagation.

%%%%%%%%%%%%%%%%%%%%%%%%%%%%%%%%%%%%%%%%%%%%%%%%%%%%%%%%%%%%%%%%%%%%%%%%%%%%%%%%%%%%%%%%%%%%%%%%%%%%%%%%%%%%%%%%%%%%%%%%%%%%%%%%%%%%%%%%%%%%%%%%%%%%%%%%%%%%
\section{Clifford  representations and associated projection operators} \label{sec:3}

To formulate the scattering problem in matrix form, we introduce a
matrix representation of the Clifford algebra. The Clifford generators
are realized by matrices satisfying the defining anti-commutation
relations, which allows the corresponding Dirac-type operator to be
represented as a finite-dimensional matrix operator.

The simplest matrix representation of  $\mathcal C\ell_{0,3}(\mathbb R)$  is given   by the Pauli matrices \cite{Lounesto-2001} 
\begin{equation} \label{sigmaMat}
    \sigma_1=
\begin{pmatrix}
0&1\\
1&0
\end{pmatrix},
\qquad
\sigma_2=
\begin{pmatrix}
0&-\im\\
\im&0
\end{pmatrix},
\qquad
\sigma_3=
\begin{pmatrix}
1&0\\
0&-1
\end{pmatrix}.
\end{equation}
These $2\times 2$ matrices provide the basic ingredient for the construction of the $\gamma$-matrices required for systems with an arbitrary number of barriers.
 
To describe a system with $N$ Dirac delta barriers, we extend this formulation to a family of Clifford generators acting on the complex vector space $\mathbb{C}^{d}$, with
\[
d=2^{\lceil N/2\rceil}.
\]
The defining anti-commutation relations are satisfied by $d\times d$ complex matrices, constructed recursively by means of the Kronecker product $\otimes$ to enlarge the representation space as $N$ grows.

For the elementary case $N=1$, the Pauli representation provides the
starting point of the recursive construction. We define
\begin{equation} \label{caseN1}
\gamma^{(1)}_0=\im \sigma_3,\qquad
\gamma^{(1)}_1=\im \sigma_1,\qquad
\gamma^{(1)}_2=\im \gamma^{(1)}_0\gamma^{(1)}_1=\sigma_2 .
\end{equation}
For $N=2$, the same two-dimensional representation is retained, namely
\[
\gamma^{(2)}_i=\gamma^{(1)}_i,\qquad i=0,1,2.
\]
For higher-order systems, $N\geq3$, the generators are constructed
recursively. In the odd-dimensional case, we define
\begin{equation*}
\gamma^{(N)}_0=\im \sigma_3\otimes\mathbbm{1}_{d/2},
\qquad
\gamma^{(N)}_{i}=\im\sigma_1\otimes\gamma^{(N-1)}_{i-1},
\qquad i=1,\dots,N,
\end{equation*}
where $\mathbbm{1}_{d/2}$ represents the identity matrix of order $d/2$.
The family of $\gamma$-matrices is completed by defining the product
matrix
\begin{equation*}
\gamma^{(N)}_{N+1}
=
\im \prod_{j=0}^{N}\gamma^{(N)}_j .
\end{equation*}
For the subsequent even value of $N$, the same representation is retained by setting
\[
\gamma^{(N+1)}_i=\gamma^{(N)}_i,
\qquad
i=0,1,\ldots,N+1.
\]

The above recursive construction ensures that, for every $N$, the matrices satisfy the Clifford anti-commutation relations
\begin{equation}\label{clif}
\{\gamma_r^{(N)},\gamma_s^{(N)}\}
=
-2\delta_{rs}\mathbbm{1}_d,
\qquad
r,s=1,2,\ldots,N.
\end{equation}
Accordingly, the matrices
$\{\gamma_r^{(N)}\}_{r=1}^{N}$
generate the Clifford algebra
$\mathcal C\ell_{0,N}$,
whereas $\gamma_0$ is distinguished and will be associated with the evolution variable in the Dirac-type formulation developed in Section~\ref{sec:4}. This representation provides the algebraic framework for constructing the Dirac operator, deriving the transfer matrix, and establishing the structural results presented in the subsequent sections.

For notational simplicity, the superscript $(N)$ and the matrix
dimension $d$ will be omitted whenever no ambiguity arises. 

The combinations of Clifford generators
\begin{equation} \label{Gamma_n}
\Gamma_n:=\gamma_0\gamma_n,\qquad n=1,2,\ldots,N,
\end{equation}
will be used to define a family of projection operators. From the
Clifford relations \eqref{clif}, these matrices satisfy
\begin{equation} \label{antiGAM}
\{\Gamma_r,\Gamma_s\}=-2\delta_{rs}\mathbbm{1}.
\end{equation}
We then define the family of operators
\begin{equation} \label{Pn}
P_n:=\frac12(\mathbbm{1}-\im \Gamma_n),
\qquad n=1,2,\ldots,N .
\end{equation}
Their algebraic properties are summarized in the following proposition.
\begin{proposition}\label{pro1}
Let $P_n$ be the operators defined by \eqref{Pn}, where
$\Gamma_n=\gamma_0\gamma_n$ and the matrices $\gamma_r$ satisfy the anti-Euclidean Clifford relations \eqref{clif}. Then the following properties hold:
\begin{enumerate}
\item[(a)] Each $P_n$ is an orthogonal projector, namely
\[
P_n^\dagger=P_n,\qquad P_n^2=P_n .
\]
\item[(b)] The anti-commutator of two projectors is given by
\begin{equation}\label{anti}
\{P_r,P_s\}
=
P_r+P_s+\frac12(\delta_{rs}-1)\mathbbm{1}.
\end{equation}
\end{enumerate}
\end{proposition}

\begin{proof} 
\begin{enumerate}
\item[(a)] 
 Since the anti-Euclidean generators are skew-Hermitian,
$\gamma_r^\dagger=-\gamma_r,$ we have
$\Gamma_n^\dagger = (\gamma_0\gamma_n)^\dagger
=
\gamma_n^\dagger\gamma_0^\dagger
=
\gamma_n\gamma_0
=
-\gamma_0\gamma_n = -\Gamma_n.$
Therefore, $(\im \Gamma_n)^\dagger = \im \Gamma_n,$ and hence
\[
P_n^\dagger
=
\frac12
\left(\mathbbm1-(\im \Gamma_n)^\dagger\right)
=
\frac12(\mathbbm1-\im \Gamma_n)
=
P_n .
\]
Moreover, using \eqref{antiGAM} we obtain $\Gamma_n^2 = - \mathbbm{1}.$ Consequently,
\[
P_n^2 =
\frac14
\left(
\mathbbm1
-2\im \Gamma_n
+\im^2 \Gamma_n^2
\right) =
\frac14
\left(
2\mathbbm1-2\im \Gamma_n
\right)
= P_n.
\]
\item[(b)]  
Using \eqref{antiGAM}, we obtain
\[
\begin{aligned}
\{P_r,P_s\}
&=
\frac14
\left(
2\mathbbm1
-2\im(\Gamma_r+\Gamma_s)
-\{\Gamma_r,\Gamma_s\}
\right)\\
&=
\frac14
\left(
2\mathbbm1
-2\im(\Gamma_r+\Gamma_s)
+2\delta_{rs}\mathbbm1
\right).
\end{aligned}
\]
Since
\[
P_r+P_s
=
\mathbbm1- \frac12 \im(\Gamma_r+\Gamma_s),
\]
we finally obtain
\[
\{P_r,P_s\}
=
P_r+P_s+\frac12(\delta_{rs}-1)\mathbbm1 ,
\]
which proves the result.
\end{enumerate}
\end{proof}

%%%%%%%%%%%%%%%%%%%%%%%%%%%%%%%%%%%%%%%%%%%%%%%%%%%%%%%%%%%%%%%%%%%%%%%%%%%%%%%%%%%%%%%%%%%%%%%%%%%%%%%%%%%%%%%%%%%%%%%%%%%%%%%%%%%%%%%%%%%%%%%%%%%%%%%%%%%%

\section{Dirac-type formulation of the Kronig-Penney model} \label{sec:4}

To connect the scalar Schrödinger equation \eqref{SCA} with the Clifford-algebraic setting of Section~\ref{sec:2}, we first rewrite \eqref{SCA} by defining the coupling parameter $\mu$ and the wave number $k$,
\[
\mu = \frac{2 m \lambda}{\hbar^2}, \qquad k = \frac{\sqrt{2 m E}}{\hbar}.
\]
In terms of these parameters, equation \eqref{SCA} takes the form
\begin{equation}\label{SCAA}
\left[-\frac{\mathrm{d}^2}{\mathrm{d}x^2}-\mu\sum_{n=1}^{N}\delta(x-nL)\right]\psi(x) = k^2\psi(x).
\end{equation}
We then factorize this second-order operator into an equivalent first-order Dirac-type system—a procedure we call \emph{spinorization}. To this end, the scalar wave function $\psi(x)$ is promoted to a spinor-valued function $\Psi(x)$ in the representation space of $\mathcal C\ell_{0,N}(\mathbb{R})$, laying the groundwork for the transfer-matrix formalism in the subsequent sections.

Since we consider a scattering problem with positive energy, the corresponding spinor field $\Psi(x)$ is generally not square-integrable over $\mathbb{R}.$ Accordingly, we work in the space $\Psi \in L^2_{\rm loc}(\mathbb{R}^+) \otimes \mathbb{C}^{2^{\lceil N/2 \rceil}},$ supplemented by the matching conditions at the points $x = nL$ generated by the Dirac delta interactions.

Consider the Dirac-type equation
\begin{equation}\label{RicG}
\mathcal{D}\Psi(x)=\kappa\Psi(x),
\qquad
\mathcal{D}:=
\gamma_0\frac{d}{dx}
-
\im \sum_{n=1}^{N}\gamma_nV_n(x),
\end{equation}
where the matrices $\gamma_n$ are constructed recursively as described in
Section~\ref{sec:3}. The corresponding spinor-valued function
$\Psi(x)$ takes values in the representation space
$\mathbb{C}^{d}$, with $d=2^{\lceil N/2\rceil}$. Applying the operator $\mathcal{D}$ once more to equation \eqref{RicG}
and using \eqref{clif}, we obtain
\begin{equation}\label{RicH}
\left[
-\frac{d^2}{dx^2}
+
\sum_{n=1}^{N}V_n^2(x)
-
\im \sum_{n=1}^{N}\gamma_0\gamma_n V_n'(x)
\right]\Psi(x)
=
\kappa^2\Psi(x).
\end{equation}
Indeed, the first-order derivative terms cancel due to the
anti-commutation relations between $\gamma_0$ and $\gamma_n$, while the mixed contributions in the square of the potential term vanish by the Clifford algebra relations. Equation \eqref{RicH} is a second-order system for the spinor-valued function $\Psi(x)$, in which the coupling between the components of the spinor is encoded in the matrices $\im \gamma_0\gamma_n$. In terms of the projection operators introduced in Proposition~\ref{pro1}, one has
\begin{equation*}
\im \gamma_0\gamma_n=\im \Gamma_n=\mathbbm{1}-2P_n .
\end{equation*}
Substituting this identity into \eqref{RicH} yields
\begin{equation}\label{RicHP}
\left[
-\frac{d^2}{dx^2}
+
\sum_{n=1}^{N}V_n^2(x)
+
\sum_{n=1}^{N}(2P_n-\mathbbm{1})V_n'(x)
\right]\Psi(x)
=
\kappa^2\Psi(x).
\end{equation}
This formulation makes explicit the algebraic role of the projectors
$P_n$ and provides the basis for the subsequent construction of a
global projection operator adapted to the structure of the interaction.
\begin{proposition} \label{pro2}Let $\{P_n\}_{n=1}^N$ be the set of projection operators defined in \eqref{Pn}. Let $\Lambda(x) := \sum_{n=1}^N V'_n(x)$ denote the sum of the potential derivatives, and define the weighted operator
\begin{equation} \label{Px}
P(x) := \frac{1}{\Lambda(x)} \sum_{n=1}^N V'_n(x) P_n.
\end{equation}
Suppose that for $N \ge 2$, the derivatives $V'_n(x)$ satisfy the orthogonality-like condition
\begin{equation} \label{cond-pro}
\sum_{r \neq s} V'_r(x) V'_s(x) = 0.
\end{equation}
Then, $P(x)$ is a projection operator.
\end{proposition}
\begin{proof} 
The operator $P(x)$ is Hermitian as a consequence of Proposition~\ref{pro1},
since it is a real linear combination of the orthogonal projection
operators $P_n$. Indeed,
\begin{equation*}
P(x)^\dagger
=
\frac{1}{\Lambda(x)}
\sum_{n=1}^{N}V_n'(x)P_n^\dagger
=
\frac{1}{\Lambda(x)}
\sum_{n=1}^{N}V_n'(x)P_n
=
P(x).
\end{equation*}
It remains to verify the idempotency property. For $N=1$, the result
follows immediately from Proposition~\ref{pro1}, since $P(x)=P_1$.
For $N\geq2$, we consider the square of $P$
\begin{equation*} \label{P2calc}
P^2 = \frac{1}{\Lambda^2} \sum_{r,s=1}^N V'_r V'_s P_r P_s = \frac{1}{2\Lambda^2} \sum_{r,s=1}^N V'_r V'_s \{P_r, P_s\},
\end{equation*}
where we have exploited the symmetry of the coefficients $V'_r V'_s$ to introduce the anti-commutator $\{P_r, P_s\}$. Substituting the identity \eqref{anti}, we obtain
\begin{equation*}
P^2 = \frac{1}{2\Lambda^2} \sum_{r,s=1}^N V'_r V'_s (P_r + P_s) + \frac{1}{4\Lambda^2}  \sum_{r,s=1}^N V'_r V'_s (\delta_{rs} - 1) \mathbb{1}.
\end{equation*}
The first summation on the right-hand side simplifies to (see \eqref{Px})
\begin{equation*}
\frac{1}{2\Lambda^2} \left[ \left( \sum_{s=1}^N V'_s \right) \sum_{r=1}^N V'_r P_r + \left( \sum_{r=1}^N V'_r \right) \sum_{s=1}^N V'_s P_s \right] = \frac{1}{2\Lambda^2} \left[ \Lambda (\Lambda P) + \Lambda (\Lambda P) \right] = P.
\end{equation*}
The second summation, involving the identity matrix, can be rewritten as
\begin{equation*}
\sum_{r,s=1}^N V'_r V'_s (\delta_{rs} -1) = \sum_{r=1}^N (V'_r)^2 - \sum_{r,s=1}^N V'_r V'_s = -\sum_{r \neq s} V'_r V'_s.
\end{equation*}
By the hypothesis \eqref{cond-pro}, this term vanishes identically. Consequently, $P^2 = P$, which completes the proof.
\end{proof}
As a consequence of the condition \eqref{cond-pro}, the second-order
equation \eqref{RicHP} simplifies to
\begin{equation*}\label{RicHnew}
\left[
-\frac{d^2}{dx^2}
+
\sum_{n=1}^{N}\left(V_n^2-V_n'\right)
\right]\Psi(x)
=
\kappa^2\Psi(x).
\end{equation*}
To reproduce the Schrödinger equation \eqref{SCAA}, we impose the
matching condition
\begin{equation}\label{cond-match}
\sum_{n=1}^{N}\left(V_n^2-V_n'\right)-\kappa^2
=
-\mu\sum_{n=1}^{N}\delta(x-nL)-k^2.
\end{equation}
A natural choice is to set $V_n'$ as a single delta interaction, i.e.,
\begin{equation}\label{derVn}
V_n'(x)=\mu\delta(x-nL).
\end{equation}
The corresponding supports are mutually disjoint for different values
of $n$, and hence the compatibility condition \eqref{cond-pro} is
automatically fulfilled. Integrating \eqref{derVn} yields the family of step potentials
\begin{equation}\label{Vn}
V_n(x) = \mu \theta(x-nL) + C_n,
\end{equation}
where $\theta$ denotes the Heaviside step function and $C_n$ are integration constants. Substituting \eqref{Vn} and \eqref{derVn} into \eqref{cond-match}, and using the property $\theta^2 = \theta$ almost everywhere, we obtain
\begin{equation*}
\sum_{n=1}^{N}
\left[
(\mu^2 + 2\mu C_n)\theta(x-nL) + C_n^2
\right]  
- \mu\sum_{n=1}^{N}\delta(x-nL) - \kappa^2 
= -\mu\sum_{n=1}^{N}\delta(x-nL) - k^2.
\end{equation*}
The $\theta$-dependent terms are eliminated by choosing
\[
C_n = -\frac{\mu}{2},
\]
which leads to
\begin{equation} \label{kappa2}
\kappa^2 = \sum_{n=1}^{N}C_n^2 + k^2 = \frac{N\mu^2}{4} + k^2.
\end{equation}
Consequently, the Schrödinger equation \eqref{SCAA} is recast as the Dirac-type system
\begin{equation}\label{DIRAC}
\left(
\gamma_0\frac{\mathrm{d}}{\mathrm{d}x}
-\mathrm{i} \sum_{n=1}^{N}\gamma_n V_n(x)
\right)\Psi(x)
=
\kappa\Psi(x),
\end{equation}
where  
\begin{equation}\label{potk-DIRAC}
V_n(x)=\mu\left[\theta(x-nL)-\frac{1}{2}\right],
\qquad
\kappa=\frac{1}{2}\sqrt{N\mu^2+4k^2}
\end{equation}
(cf. \eqref{RicG}, \eqref{Vn}, and \eqref{kappa2}).
This construction provides a factorization of the second-order
Schrödinger operator into a first-order Clifford-valued system, in
close analogy with the factorization techniques of supersymmetric
quantum mechanics; see, for example, \cite{Cooper-2001,Thaller-1992}.

%%%%%%%%%%%%%%%%%%%%%%%%%%%%%%%%%%%%%%%%%%%%%%%%%%%%%%%%%%%%%%%%%%%%%%%%%%%%%%%%%%%%%%%%%%%%%%%%%%%%%%%%%%%%%%%%%%%%%%%%%%%%%%%%%%%%%%%%%%%%%%%%%%%%%%%%%%%%

\section{Clifford-based exact solutions and scattering properties} \label{sec:5}

In this section, we develop the scattering analysis for the Dirac-type equation \eqref{DIRAC} associated with the finite array of localized potentials introduced above. The piecewise-constant nature of the interaction allows the longitudinal domain to be decomposed into subregions where the differential operator has constant coefficients. This   leads to a construction in terms of local propagators, whose composition yields the global transfer matrix of the system. The resulting Clifford-algebraic formulation provides both an explicit representation of the spinor-valued function and a direct characterization of the corresponding scattering coefficients.

The discontinuity points of the potential profile in \eqref{potk-DIRAC} induce the decomposition of the domain into the intervals
\begin{equation} \label{intervalli}
I_0 = (-\infty, x_1], \quad
I_\ell = [x_\ell,x_{\ell+1}], \quad
I_N=[x_N,+\infty),
\qquad x_\ell=\ell L,\qquad \ell=1,2,\dots,N-1 .
\end{equation}
Since the potentials are constant on each interval $I_r$ for $r=0,1,\dots,N$, equation \eqref{DIRAC} reduces to 
\begin{equation*}\label{eqg0Mr2}
	\gamma_0 \Psi'_r(x)-\im M_r\Psi_r(x)=\kappa\Psi_r(x),
	\qquad
	M_r=\sum_{n=1}^N\gamma_n V_n^{(r)},
	\qquad
	V_n^{(r)}=\mu\Bigl[\theta(rL-nL)-\frac12\Bigr],
\end{equation*}
where  $\Psi_r (x) \equiv \Psi(x)|_{I_r}.$
Multiplying from the left by \((-\gamma_0)\), we obtain the equivalent first-order system 
\begin{equation}\label{Hr}
	\Psi'_r(x)=H_r\Psi_r(x),
	\qquad
	H_r=-\gamma_0(\kappa\mathbbm{1}+\im M_r).
\end{equation}
Since 
\begin{equation}\label{Mr2}
	M_r=\frac{\mu}{2}\sum_{n=1}^N s_n^{(r)}\gamma_n,
	\qquad
	s_n^{(r)}=
	\begin{cases}
	+1, & n=1,\dots,r,\\
	-1, & n=r+1,\dots,N,
	\end{cases}
\end{equation}
is constant on \(I_r\), the matrix \(H_r\) is constant as well. Therefore, the local solution is given by
\begin{equation*}\label{solexp}
	\Psi_r(x)=\exp\bigl((x-x_r)H_r\bigr)\Psi_r(x_r).
\end{equation*}
At each interface \(x_{r+1}\), \(r=0,1,\dots,N-1\), continuity of the spinor field imposes
\[
\Psi_r(x_{r+1})=\Psi_{r+1}(x_{r+1}).
\]
Evaluating the local solutions at two consecutive interfaces, we obtain the matching relation
\begin{equation}\label{eq:matching}
	\Psi_r(x_{r})
	=
	e^{-LH_r}\Psi_{r+1}(x_{r+1}),
\end{equation}
where \(L=x_{r+1}-x_r\) denotes the length of each interval. By iterating the above relation along the chain, the spinors in the asymptotic regions are connected by
\[
\Psi_0(x_0)=\mathcal{T}_N\Psi_N(x_N),
\]
where the global transfer operator is given by
\begin{equation*}\label{eq:global_T}
	\mathcal{T}_N
	=
	e^{-LH_0} e^{-LH_1} \cdots e^{-LH_{N-1}}.
\end{equation*}
%%%%%%%%%%%%%%%%%%%%%%%%%%%%%%%%%%%%%%%%%%%%%%
\paragraph{Elementary cases: $N=1$ and $N=2$.}
For completeness, we illustrate the previous construction in the two simplest
configurations, corresponding to one and two localized interactions. These
examples provide explicit realizations of the wave function and clarify the role of the Clifford generators in the construction of
the transfer matrix. In both cases, the spinor space is two-dimensional and
the Clifford algebra admits a representation in terms of Pauli matrices,
consistently with the convention introduced in Section~\ref{sec:3}.

For a single localized potential, $N=1$, the domain is divided into two
regions, namely $I_0=(-\infty,x_1]$ and $I_1=[x_1,+\infty)$. The spinor field
takes the standard scattering form
\begin{equation*}
\Psi_0(x)=\mathcal A_{0+}e^{\im kx}
+\mathcal A_{0-}e^{-\im kx},
\qquad x\in I_0 ,
\end{equation*}
and
\begin{equation*}
\Psi_1(x)=\mathcal A_{1+}e^{\im kx},
\qquad x\in I_1 ,
\end{equation*}
where $\mathcal A_{0+}$ and $\mathcal A_{0-}$ denote the amplitudes of the
incident and reflected components, respectively, while $\mathcal A_{1+}$ is
the transmitted amplitude. The absence of the reflected component in
$I_1$ implements the outgoing boundary condition in the transmitted region. Substitution of these expressions into \eqref{Hr} gives the
algebraic conditions
\begin{eqnarray*}
    \left(
\pm \im k\mathbbm 1
-H_0
\right)\mathcal A_{0\pm}=0, \qquad
\left(
\im k\mathbbm 1
-H_1
\right)\mathcal A_{1+}=0.
\end{eqnarray*}
A non-trivial solution exists only if the corresponding matrices have vanishing determinant. 
This occurs when
\begin{equation*}
\kappa^2 = k^2 + \frac{\mu^2}{4},
\end{equation*}
reproducing the formula for $\kappa$ in \eqref{potk-DIRAC} when $N=1$. The remaining coefficients are fixed by continuity of the spinor fields at the
interface,
\begin{equation*}
\Psi_0(x_1)=\Psi_1(x_1).
\end{equation*}
Equivalently,
\begin{equation*}
\mathcal A_{0+}e^{\im kx_1}
+\mathcal A_{0-}e^{-\im kx_1}
=
\mathcal A_{1+}e^{\im kx_1},
\end{equation*}
which determines the transmitted amplitude in terms of the incident and
reflected components. In the present formalism, this relation is the
one-barrier realization of the general interface matching condition
\eqref{eq:matching}.

The case $N=2$ provides the first non-trivial example involving an internal
region. The three intervals
\[
I_0=(-\infty,x_1],\qquad
I_1=[x_1,x_2],\qquad
I_2=[x_2,+\infty)
\]
lead to the decomposition
\begin{align*}
\Psi_0(x)&=
\mathcal A_{0+}e^{\im kx}
+\mathcal A_{0-}e^{-\im kx},
\\
\Psi_1(x)&=
\mathcal A_{1+}e^{\im kx}
+\mathcal A_{1-}e^{-\im kx},
\\
\Psi_2(x)&=
\mathcal A_{2+}e^{\im kx}.
\end{align*}
In each interval, the amplitudes satisfy the same local algebraic equation
\begin{equation*}
\left(
\pm \im k\mathbbm 1 -H_r
\right)\mathcal A_{r\pm}=0, \quad r=0,1, \qquad
\left(
\im k\mathbbm 1 -H_2
\right)\mathcal A_{2+}=0,
\end{equation*}
where the sign configurations associated with the three regions are (see \eqref{Mr2})
\begin{equation*}
(s_1^{(0)}, s_2^{(0)})=(-1,-1),\qquad
(s_1^{(1)}, s_2^{(1)})=(1,-1),\qquad
(s_1^{(2)}, s_2^{(2)})=(1,1).
\end{equation*}
Again, the solvability condition is obtained from the vanishing determinant of the local algebraic system. The coefficients in the three regions are finally related through the
continuity conditions
\begin{equation*}
\Psi_0(x_1)=\Psi_1(x_1),
\qquad
\Psi_1(x_2)=\Psi_2(x_2),
\end{equation*}
which reproduce explicitly the composition of local propagators  and give the transfer matrix.

The transfer representation obtained above encodes the complete propagation
through the sequence of localized interactions. We now establish its
connection with the physical scattering observables by analyzing the Clifford
bilinear form associated with the spinor field. This quantity provides a
natural characterization of the transmitted and reflected components and
allows us to derive the corresponding scattering coefficients.

We introduce the Clifford bilinear form
\[
\rho(x)=\Psi^\dagger(x)\gamma_0\Psi(x),
\]
induced by the Clifford representation given in Section~\ref{sec:3}, where the generator $\gamma_0$  satisfies the algebraic relation \eqref{clif}. This generator defines the bilinear structure naturally associated with the chosen Clifford representation. Differentiating $\rho(x)$ with respect to $x$ and using the Dirac-type equation \eqref{Hr}, we obtain
\begin{equation*}
\frac{d}{dx}\rho(x) = 0.
\end{equation*}
 This equation implies the conservation of the probability from the incoming region $(-\infty, x_1]$ to the outgoing region  $[x_N, +\infty)$ (recall that for the scalar Schr\"odinger equation, the conserved probability is given by $\rho_S = 
\frac12 \left(\frac{d\psi^\dagger}{dx} \psi - 
\psi^\dagger \frac{d\psi}{dx}
\right)$).

%%%%%%%%%%%%%%%%%%%%%%%%%%%%%%%%%%%%%%%%%%%%%%%%%%%%%%%%%%%%%%%%%%%%%%%%%%%%%%%%%%%%%%%%%%%%%%%%%%%%%%%%%%%%%%%%%%%%%%%%%%%%%%%%%%%%%%% 
\section{Generalized Cauchy–Riemann formulation and Helmholtz equations} \label{sec:6}

In this section, we show that the Dirac-type equation introduced in \eqref{DIRAC} can be derived from a higher-dimensional relativistic Dirac equation. To this end, we consider the time-independent Dirac equation in $N+1$ spatial dimensions
\begin{equation} \label{anAx}
\left(  \gamma_0 \frac{\partial}{\partial x} +  \sum_{n=1}^N \gamma_n \frac{\partial}{\partial y_n} + \im \mu \gamma_0 \sum_{n=1}^N y_n \delta(x - n L) \right) \Psi(x, {\bf y} ) = \kappa \Psi(x, {\bf y} ),
\end{equation}
where ${\bf y}=(y_1,y_2,\dots,y_N)$ represents the set of transverse coordinates. By means of the gauge transformation
\begin{equation}\label{anB}
\Psi(x,\mathbf{y}) = \exp\left(-\mathrm{i}\sum_{n=1}^{N} y_n V_n(x)\right)\Psi(x),
\end{equation}
with $V_n(x)$ defined as in \eqref{potk-DIRAC}, the stationary equation \eqref{anAx} reduces to the one-dimensional Dirac-type equation \eqref{DIRAC}. Indeed, the transverse dependence is entirely absorbed by the phase factor
\[
\sum_{n=1}^{N}\gamma_n\frac{\partial\Psi(x,\mathbf{y})}{\partial y_n} = -\mathrm{i}\sum_{n=1}^{N}\gamma_n V_n(x)\Psi(x,\mathbf{y}).
\]
Furthermore, differentiating the exponential with respect to $x$ by means of \eqref{derVn} gives
\[
\gamma_0 \left[\frac{\partial}{\partial x} \exp\left(-\mathrm{i}\sum_{n=1}^{N} y_n V_n(x)\right)\right] \Psi(x) = -\mathrm{i}\mu\gamma_0 \sum_{n=1}^{N}y_n\delta(x-nL)\Psi(x,\mathbf{y}),
\]
which cancels the singular delta interactions in \eqref{anAx}. The transformed spinor $\Psi(x)$ thus obeys
\begin{equation*}
\left(\gamma_0\frac{\mathrm{d}}{\mathrm{d}x} - \mathrm{i}\sum_{n=1}^{N}\gamma_n V_n(x)\right)\Psi(x) = \kappa\Psi(x),
\end{equation*}
recovering equation \eqref{DIRAC}. This demonstrates how the reduced Dirac-type model emerges directly from the higher-dimensional relativistic framework, linking it to the first-order Kronig-Penney model.

Restricting equation \eqref{anAx} to the intervals $I_r$ defined in \eqref{intervalli}, where the singular interactions vanish, the spinor field $\Psi_r(x,\mathbf{y}) \equiv \Psi(x,\mathbf{y})|_{I_r}$ obeys the free Dirac equation
\begin{equation}\label{appeA_fixed}
\left(
\gamma_0\frac{\partial}{\partial x}
+\sum_{n=1}^{N}\gamma_n\frac{\partial}{\partial y_n}
\right)
\Psi_r(x,\mathbf{y})
=
\kappa\Psi_r(x,\mathbf{y}).
\end{equation}
Multiplying \eqref{appeA_fixed} from the left by \(-\gamma_0/2\), and using
the Clifford relation \(\gamma_0^2=-\mathbb{1}\), we obtain
\begin{equation*}
\frac12
\left(
\frac{\partial}{\partial x}
-\sum_{n=1}^{N}\Gamma_n
\frac{\partial}{\partial y_n}
\right)
\Psi_r(x,\mathbf{y})
=
-\frac{\kappa}{2}\gamma_0\Psi_r(x,\mathbf{y}),
\end{equation*}
where the Clifford generators \(\Gamma_n\) are defined in \eqref{Gamma_n}.
Thus, introducing the conjugate generalized Cauchy-Riemann operator (see \eqref{dx} and \eqref{del-operator})
\[
D:=
\frac12
\left(
\frac{\partial}{\partial x}
-\sum_{n=1}^{N}\Gamma_n
\frac{\partial}{\partial y_n}
\right),
\]
the Dirac equation can be written as
\begin{equation}\label{CR_system}
D\Psi_r(x,\mathbf{y})
=
-\frac{\kappa}{2}\gamma_0\Psi_r(x,\mathbf{y}).
\end{equation}
Using the conjugation property
\[
\gamma_0D\gamma_0=-\bar D,
\]
we rewrite \eqref{CR_system} as a coupled generalized
Cauchy--Riemann system. Setting
\[
\widetilde{\mathcal A}(x,\mathbf{y})
=\Psi_r(x,\mathbf{y}),
\qquad
\widetilde{\mathcal B}(x,\mathbf{y})
=\gamma_0\Psi_r(x,\mathbf{y}),
\]
we obtain
\begin{equation}\label{newA}
\left\{
\begin{aligned}
D\widetilde{\mathcal A}(x,\mathbf{y})
&=-\frac{\kappa}{2}
\widetilde{\mathcal B}(x,\mathbf{y}),
\\[2mm]
\bar D\widetilde{\mathcal B}(x,\mathbf{y})
&=\frac{\kappa}{2}
\widetilde{\mathcal A}(x,\mathbf{y}).
\end{aligned}
\right.
\end{equation}

This formulation allows the first-order Dirac system to be decoupled into two
second-order equations. Applying \(\bar D\) and \(D\) from the left to the
first and second equations in \eqref{newA}, respectively, yields
\begin{equation*}\label{newB}
\left\{
\begin{aligned}
\bar D D\,\widetilde{\mathcal A}(x,\mathbf{y})
&=
-\frac{\kappa^2}{4}
\widetilde{\mathcal A}(x,\mathbf{y}),
\\[2mm]
D\bar D\,\widetilde{\mathcal B}(x,\mathbf{y})
&=
-\frac{\kappa^2}{4}
\widetilde{\mathcal B}(x,\mathbf{y}).
\end{aligned}
\right.
\end{equation*}
By introducing the characteristic Clifford variables
\[
Z=x+ \frac{1}{N}\sum_{n=1}^{N}\Gamma_n y_n,
\qquad
\bar Z=x- \frac{1}{N} \sum_{n=1}^{N}\Gamma_n y_n,
\]
the generalized Cauchy-Riemann operators admit the identification
\[
D=\partial_{\bar Z},
\qquad
\bar D=\partial_Z .
\]
Hence, setting
\[
\mathcal A(Z,\bar Z):=\widetilde{\mathcal A}(x,\mathbf y),
\qquad
\mathcal B(Z,\bar Z):=\widetilde{\mathcal B}(x,\mathbf y),
\]
the system \eqref{newA} takes the form
\begin{equation}\label{newA-Z}
\left\{
\begin{aligned}
\partial_{\bar Z}\mathcal A(Z,\bar Z)
&=
-\frac{\kappa}{2}\mathcal B(Z,\bar Z),
\\[2mm]
\partial_Z\mathcal B(Z,\bar Z)
&=
\frac{\kappa}{2}\mathcal A(Z,\bar Z).
\end{aligned}
\right.
\end{equation}
Applying \(\partial_Z\) and \(\partial_{\bar Z}\) to the first and second
equations, respectively, and using \eqref{newA-Z}, we obtain the decoupled
second-order equations
\begin{equation}\label{newB-Z}
\left\{
\begin{aligned}
\partial_Z\partial_{\bar Z}\mathcal A(Z,\bar Z)
&=
-\frac{\kappa^2}{4}\mathcal A(Z,\bar Z),
\\[2mm]
\partial_{\bar Z}\partial_Z\mathcal B(Z,\bar Z)
&=
-\frac{\kappa^2}{4}\mathcal B(Z,\bar Z).
\end{aligned}
\right.
\end{equation}
\color{black}
These equations provide the generalized Helmholtz formulation associated with
the Clifford-valued fields \(\mathcal A\) and \(\mathcal B\). Thus, the
original first-order Dirac-type system decomposes into two independent second-order problems in the corresponding spinorial spaces.

%%%%%%%%%%%%%%%%%%%%%%%%%%%%%%%%%%%%%%%%%%%%%%%%%%%%%%%%%%%%%%%%%%%%%%%%%%%%%%%%%%%%%%%%%%%%%%%%%%%%%%%%%%%%%%%%%%%%%%%%%%%%%%%%%%%%%%%%%%%%%%%%%%%%%%%%%%%% 
\section{Appell expansions for Clifford-valued Helmholtz equations} \label{sec:7}

To construct explicit solutions of the Dirac system \eqref{newA-Z}, we
introduce the bivariate Appell polynomial family
$\{\mathcal P_{l,m}(Z,\bar Z)\}_{l,m\in\mathbb N_0},$
whose structure is adapted to the pair of conjugate operators
\(\partial_{\bar Z}\) and \(\partial_Z\). The polynomials satisfy the
following Appell-type relations
\begin{equation}\label{mewK}
\begin{aligned}
\partial_{\bar Z}\mathcal P_{l,m}(Z,\bar Z)
&=
l\,\mathcal P_{l-1,m}(Z,\bar Z),
\\[2mm]
\partial_Z\mathcal P_{l,m}(Z,\bar Z)
&=
m\,\mathcal P_{l,m-1}(Z,\bar Z).
\end{aligned}
\end{equation}
The action of the second-order operators appearing in the generalized Helmholtz equations \eqref{newB-Z} follows directly from these relations. In particular, one has
\begin{equation*}\label{eq:second_order_action}
\begin{aligned}
\partial_{\bar Z}\partial_Z
\mathcal P_{l,m}(Z,\bar Z)
&=
\partial_{\bar Z}
\left(
m\,\mathcal P_{l,m-1}(Z,\bar Z)
\right)
\\
&=
lm\,\mathcal P_{l-1,m-1}(Z,\bar Z),
\end{aligned}
\end{equation*}
while, analogously,
\[
\partial_Z\partial_{\bar Z}
\mathcal P_{l,m}(Z,\bar Z)
=
lm\,\mathcal P_{l-1,m-1}(Z,\bar Z).
\]
Therefore, the Appell family provides a natural polynomial basis adapted to the first-order Dirac system.  Expanding the Clifford-valued fields in the Appell basis
\begin{equation*}\label{eq:Appell_expansion}
\mathcal A(Z,\bar Z)
=
\sum_{l,m\geq0}a_{l,m}\mathcal P_{l,m}(Z,\bar Z),
\qquad
\mathcal B(Z,\bar Z)
=
\sum_{l,m\geq0}b_{l,m}\mathcal P_{l,m}(Z,\bar Z),
\end{equation*}
and substituting these expansions in 
\eqref{newA-Z}, we obtain the coefficient relations
\begin{equation*}\label{eq:Appell_recurrence}
(l+1)a_{l+1,m}
=
-\frac{\kappa}{2}b_{l,m},
\qquad
(m+1)b_{l,m+1}
=
\frac{\kappa}{2}a_{l,m},
\end{equation*}
which immediately imply
\begin{equation*}\label{eq:recurrence_relation}
(l+1)(m+1)a_{l+1,m+1}
=
-\frac{\kappa^2}{4}a_{l,m},
\qquad l,m\geq0 .
\end{equation*}
The conservation of the index difference $s=l-m\in\mathbb Z$ decomposes the solution space into independent lattice diagonals. Given boundary values $\mathbf c_s, \mathbf d_s \in \mathcal C\ell_{0,N}$, the coefficients evolve as
\begin{equation*}
a_{n+s,n} = \frac{(-\kappa^2/4)^n}{n!(n+s)!}\mathbf c_s \quad (s\ge0),
\qquad
a_{n,n+s} = \frac{(-\kappa^2/4)^n}{n!(n+s)!}\mathbf d_s \quad (s>0).
\end{equation*}
Superimposing these modes over $s$ yields the full expansion for $\mathcal A(Z,\bar Z)$:
\begin{equation*}\label{eq:A_explicit_solution}
\mathcal A(Z,\bar Z)
=
\sum_{s=0}^{\infty}
\left[
\sum_{n=0}^{\infty}
\frac{\left(-{\kappa^2}/{4}\right)^n}
{n!(n+s)!}
\mathcal P_{n+s,n}(Z,\bar Z)
\right]\mathbf c_s
+
\sum_{s=1}^{\infty}
\left[
\sum_{n=0}^{\infty}
\frac{\left(-{\kappa^2}/{4}\right)^n}
{n!(n+s)!}
\mathcal P_{n,n+s}(Z,\bar Z)
\right]\mathbf d_s .
\end{equation*}
The companion field \(\mathcal B(Z,\bar Z)\) is then determined by the first-order coupling in \eqref{newA-Z}. Namely,
\[
\mathcal B(Z,\bar Z)
=
\frac{2}{\kappa}
\partial_{\bar Z}\mathcal A(Z,\bar Z).
\]
Using the first relation in \eqref{mewK} we obtain
\begin{equation*}\label{eq:B_explicit_solution}
\mathcal B(Z,\bar Z)
=
-\frac{2}{\kappa}
\sum_{s=0}^{\infty}
\left[
\sum_{n=0}^{\infty}
\frac{\left(-{\kappa^2}/{4}\right)^n}
{n!(n+s)!}
(n+s)\,
\mathcal P_{n+s-1,n}(Z,\bar Z)
\right]\mathbf c_s
-
\frac{2}{\kappa}
\sum_{s=1}^{\infty}
\left[
\sum_{n=0}^{\infty}
\frac{\left(-{\kappa^2}/{4}\right)^n}
{n!(n+s)!}
n\,
\mathcal P_{n-1,n+s}(Z,\bar Z)
\right]\mathbf d_s.
\end{equation*}
\begin{remark}
For \(m=0\), the bivariate Appell family reduces to the family of hypercomplex Appell polynomials in Clifford analysis defined by \eqref{appell-pol}; see, for example,
\cite{FalcaoMalonek2007,Malonek1990,MalonekFalcao2007}.
\end{remark}

\paragraph{Elementary case: \(N=1\).}
For $N=1,$ the Clifford variables coincide with the standard complex coordinates
\[
Z=z,\qquad \bar Z=\bar z ,
\]
and the bivariate Appell family simplifies to the classical bivariate basis
\[
\mathcal P_{l,m}(z,\bar z)=z^m\bar z^l .
\]
Accordingly, the operators \(D\) and \(\bar D\) become the Cauchy--Riemann operators
\[
D=\partial_{\bar z},\qquad \bar D=\partial_z ,
\]
and the generalized Helmholtz system specializes to its classical two-dimensional complex counterpart. Consequently, the Clifford--Appell expansion coincides with the standard Appell expansion in the complex plane.

%%%%%%%%%%%%%%%%%%%%%%%%%%%%%%%%%%%%%%%%%%%%%%%%%%%%%%%%%%%%%%%%%%%%%%%%%%%%%%%%%%%%%%%%%%%%%%%%%%%%%%%%%%%%%%%%%%%%%%%%%%%%%%%%%%%%%%%%%%%%%%%%%%%%%%%%%%%%

\section{Concluding remarks and future developments} \label{sec:FINE}
 
In this work, we have established a novel connection between the Dirac-type systems and the combinatorial structures associated with hypercomplex Appell polynomials. By interpreting the Schrödinger equation of the finite Kronig-Penney model as a consistency condition for a first-order Dirac operator, we have shown that the corresponding scattering problem admits a natural formulation within the framework of Clifford analysis.

Furthermore, the dimensional reduction from the $(N+1)$-dimensional Dirac equation to the effective one-dimensional scattering problem, achieved via a suitable gauge transformation, provides a physical explanation for the emergence of generalized bivariate Appell polynomials.

Beyond its analytical scope, the present formulation also offers promising perspectives for future numerical developments. By decoupling the Clifford-algebraic structure from the longitudinal dynamics, the hypercomplex representation of the transfer matrix provides a suitable framework for structure-preserving numerical schemes. Future work will focus on discretization strategies for the longitudinal evolution operator, leveraging the local propagator representation to construct efficient approximations of the global transfer matrix for general non-periodic and multi-layered potentials.

\vskip3pt

%%%%%%%%%%%%%%%%%%%%%%%%%%%%%%%%%%%%%%%%%%%%%%%%%%%%%%%%%
\subsection*{Declaration of competing interest}

No potential conflict of interest was reported by the authors.

\vskip3pt

%%%%%%%%%%%%%%%%%%%%%%%%%%%%%%%%%%%%%%%%%%%%%%%%%%%%%%%%%
\subsection*{Acknowledgments}
The first author is a member of the \textit{Gruppo Nazionale per il Calcolo Scientifico} of the Istituto Nazionale di Alta Matematica.
Her research was partially supported by the  INdAM--GNCS Project, code CUP$\_$E53C25002010001.
The work of the third author at CIDMA was supported by FCT (Portuguese Foundation for Science and Technology) under Projects 
\noindent UID/04106/2025 (https://doi.org/10.54499/UID/04106/2025) and 
\noindent UID/PRR/04106/2025 \\(https://doi.org/10.54499/UID/PRR/04106/2025).

%%%%%%%%%%%%%%%%%%%%%%%%%%%%%%%%%%%

% \bibliographystyle{plain}
% \bibliography{bibliography}

% %\bibliographystyle{model1-num-names}
% %\bibliographystyle{cas-model2-names}
% % Loading bibliography database
% \bibliography{cas-references}

\begin{thebibliography}{99}

\bibitem{AcetoGrassiMalonek2026} Aceto, L.,  Grassi, P.A., Malonek, H.R.. 2026. Wigner numbers and combinatorial structures of the transfer matrix in multiple Dirac delta scattering, submitted.


\bibitem{Wigner-2019}
Allen, W.D., 2019. Wigner numbers. J. Chem. Phys. {\bf 151}, 244122. 
	

\bibitem{CommentWigner-2021}
\'{A}rend\'{a}s, P.,  Cs\'{a}sz\'{a}r, A.G., 2021. Comment on ``Wigner numbers".  J. Chem. Phys. {\bf 154}, 087101. 

\bibitem{Bock-2014}
Bock, S., 2014. On Monogenic Series Expansions with Applications to Linear Elasticity Advances in Applied Clifford Algebras 24, 931--943.


\bibitem{BrackxDelangheSommen-1982}
Brackx, F., Delanghe, R., Sommen, F., 1982. Clifford analysis. Pitman, Boston-London-Melbourne.
	 
\bibitem{BSL-1992}
Brackx, F., De Schepper, H., Lávička, R. et al. 2015. Embedding Factors for Branching in Hermitian Clifford Analysis. Complex Anal. Oper. Theory 9, 355--378.

	 
\bibitem{CMF-2019}
Caçao, I., Falcão, M.I., Malonek, H.R., 2019. On generalized Vietoris’ number sequences. Discrete Appl. Math. {\bf 269}, 77--85.

\bibitem{CacaoMalonek2008}
Ca{\c{c}}{\~a}o, I., Falc{\~a}o,  M.L.,  Malonek, H.R., 2008. On the construction of special monogenic polynomials,
Int. J. Math. Math. Sci., \textbf{2008}, Art. ID 240267, 10 pp.

\bibitem{Cacao-2024}
Caçao, I., Falcão, M.I., Malonek, H.R., Tomaz G., 2024. A Sturm-Liouville equation on the crossroads of continuous and discrete hypercomplex analysis. Math. Meth. Appl. Sci. \textbf{47} , 7962--7987.

\bibitem{CMFT-2026}
Cação, I., Falcão, M.I.,  Malonek, H.R., Tomaz, G., 2026. Pascal Trapezoids as Wigner Numbers and Some Combinatorial Properties.  In: Gervasi, O., et al. Computational Science and Its Applications – ICCSA 2025 Workshops. ICCSA 2025. Lecture Notes in Computer Science, vol 15887. Springer, Cham. 

\bibitem{CMFT2018}
Ca\c{c}\~{a}o, I.,  Malonek, H.R., Falc\~{a}o, M.I., Tomaz, G. , 2018. Combinatorial identities associated with a multidimensional polynomial sequence. J. Integer Sequences {\bf 21}, Article 18.7.4.
	
	
\bibitem{Cacao-2023} 
Cação, I., Malonek, H.R., Falcão, M.I., Tomaz, G., 2023. Intrinsic properties of a non-symmetric number triangle. J. Integer Seq., \textbf{26}(4): Art. 23.4.8, 12.
	
\bibitem{Cooper-2001}
Cooper, F.A., Khare, A., Sukhatme. U., 2001. Supersymmetry in Quantum Mechanics.  World Scientific Publishing Company.
  
	
\bibitem{CCFMT-2025}
Cruz, C., Falcão, M.I., Malonek, H.R. et al. 2025. On the Pascal simplex with hypercomplex entries. Complex Anal. Oper. Theory 19, 186.

\bibitem{Eelbode-2012}
Eelbode, D., 2012. Monogenic Appell Sets as Representations of the Heisenberg Algebra. Adv. Appl. Clifford Algebras 22, 1009--1023.
	
\bibitem{Falcao-2006}
Falcão, M.I., Cruz, J., Malonek, H.R., 2006.
	\newblock Remarks on the generation of monogenic functions,
	\newblock In: K. G\"{u}rlebeck and C. K\"{o}nke, (eds.),
	\newblock Proc. of the 17-th Inter. Conf. on the Application of Computer Science and Mathematics in Architecture and
	Civil Engineering, Bauhaus-University Weimar, ISSN 1611--4086.
	
\bibitem{Falcao-2012} 
Falcão, M.I., Malonek., H.R., 2012. A note on a one-parameter family of non-symmetric number triangles. Opuscula Math., \textbf{32}(4): 661--673.

\bibitem{FalcaoMalonek2007}
Falc\~ao, M.I.,    Malonek, H.R., 2007. Generalized exponentials through Appell sets in $\mathbb{R}^{n+1}$
and Bessel functions,
in:
T. E. Simos, G. Psihoyios, C. Tsitouras (eds.),
Numerical Analysis and Applied Mathematics,
AIP Conference Proceedings, Vol.~936,
American Institute of Physics, Melville, NY,   738--741.


\bibitem{FGS-2025} Figueroa, J., Gonzalez, I., Salinas-Arizmendi. D., 2025. A novel transfer matrix framework for multiple Dirac delta potentials. Phys. Lett. A, \textbf{555 }:Paper No. 130785, 7.
	
\bibitem{Gurlebeck-1999}
G{\"u}rlebeck, K., Malonek, H.R., 1999. A hypercomplex derivative of monogenic functions in ${\mathbb{R}}^{n+1}$ and its applications. Complex Variables Theory Appl. \textbf{39}, 199--228.
	
\bibitem{Hitzer-2017}
Hitzer, E., Sangwine, St., 2017. Multivector and multivector matrix inverses in real Clifford algebras,
Applied Mathematics and Computation, Vol. 311, 375--389. 
	

\bibitem{Lounesto-2001}
Lounesto P., 2001. Clifford Algebras and Spinors, London Mathematical Society Lecture Note Series 286.  Cambridge University Press. 


\bibitem{Malonek1990}
Malonek, H.R., 1990. A new hypercomplex structure of the Euclidean space $\mathbb{R}^{m+1}$
and the concept of hypercomplex differentiability,
Complex Variables
\textbf{14}, 25--33.


\bibitem{Malonek-Falcao-2006}
Malonek, H.R., Falcão, M.I., 2006.
3D-mappings using monogenic functions, [in:] Th.E. Simos et al., eds, ICNAAM-2006 Conference Proceedings, Wiley-VCH, Weinheim, 2006, 615–-619.


\bibitem{MalonekFalcao2007}
Malonek, H.R.,  Falc\~ao,  M.I., 2007.
Special monogenic polynomials--properties and applications,
in:
T. E. Simos, G. Psihoyios, C. Tsitouras (eds.),
Numerical Analysis and Applied Mathematics,
AIP Conference Proceedings, Vol.~936,
American Institute of Physics, Melville, NY,  764--767.

\bibitem{Thaller-1992}
Thaller, B., 1992. The Dirac Equation. Texts and Monographs in Physics. Springer-Verlag Berlin Heidelberg.


\end{thebibliography}

%\end{document}

%%%%%%%%%%%%%%%%%%%%%%%%%%%%%%%%%%%

\end{document}